\documentclass[letterpaper, 10 pt, conference]{ieeeconf}  

\IEEEoverridecommandlockouts                              

\usepackage{amsmath} 
\usepackage{graphicx}
\usepackage{float}
\usepackage{cite}
\usepackage{pgfplots}
\usepackage{pgfplotstable}
\usepackage{booktabs} 
\usepackage{multirow} 
\usepackage{amssymb}
\usepackage[acronym]{glossaries}
\usepackage{bm}

\usepackage[section]{placeins}

\newtheorem{theorem}{Theorem}
\newtheorem{definition}{Definition}
\newtheorem{problem}{Problem}
\newtheorem{proposition}{Proposition}

\usepackage{array} 

\DeclareMathOperator{\R}{\mathbb R}
\DeclareMathOperator*{\argmin}{arg\,min}

\newacronym{cbf}{CBF}{control barrier function}

\title{\LARGE \bf
Belief Space Control of Safety-Critical Systems Under State-Dependent Measurement Noise
}

\author{Rohan Walia, Mitchell Black, Andrew Schoer, Kevin Leahy
\thanks{DISTRIBUTION STATEMENT A. Approved for public release. Distribution is unlimited.}
\thanks{R. Walia and K. Leahy are with the Robotics Engineering Department, Worcester Polytechnic Institute, Worcester, MA, USA.
        {\tt\small \{rwalia,kleahy\} @wpi.edu}}%
\thanks{M. Black and A. Schoer are with MIT Lincoln Laboratory, Lexington, MA, USA.
        {\tt\small \{mitchell.black,andrew.schoer\} @ll.mit.edu}}%
}

\begin{document}

\maketitle
\thispagestyle{empty}
\pagestyle{empty}

\begin{abstract}
Safety-critical control is imperative for deploying autonomous systems in the real world. Control Barrier Functions (CBFs) offer strong safety guarantees when accurate system and sensor models are available. However, widely used additive, fixed-noise models are not representative of complex sensor modalities with state-dependent error characteristics. Although CBFs have been designed to mitigate uncertainty using fixed worst-case bounds on measurement noise, this approach can lead to overly-conservative control. To solve this problem, we extend the Belief Control Barrier Function (BCBF) framework to accommodate state-dependent measurement noise via the Generalized Extended Kalman Filter (GEKF) algorithm, which models measurement noise as a linear function of the state. Using the original BCBF framework as baseline, we demonstrate the performance of the BCBF-GEKF approach through simulation results on a 1D single integrator setpoint tracking scenario and 2D unicycle kinematics trajectory tracking scenario. Our results confirm that the BCBF-GEKF approach offers less conservative control with greater safety.
\end{abstract}


\section{Introduction}
\label{sec:intro}

Safe control is a critical requirement for the deployment of autonomous systems in the real world. Control barrier functions (CBFs) are being increasingly adopted as a minimally invasive and flexible approach to enforce safety constraints for control applications~\cite{cbf_theory_and_applications}. CBFs have been successfully deployed across a wide range of systems and scenarios, such as adaptive cruise control~\cite{CBF_adaptive_cruise_control}, geofencing for fixed-wing aircraft~\cite{cbf_geofencing}, and docking of spacecraft~\cite{cbf_spacecraft_docking}.

CBFs provide strong safety assurances when a complete, accurate system model is available. 
In practice, however, sensor noise or disturbances introduce uncertainty about the true state of the system. Recent work has focused on developing techniques for applying CBFs in stochastic settings~\cite{FxTS_RaCBF, RCBF_uncertainty_bound, agrawal2022safe, garg2024advances} and hybrid systems \cite{cbf_hybrid_systems}. 

One increasingly prevalent sensing modality in robotic and autonomous systems is learning-based vision systems, which belong to the broader class of learning-enabled components (LECs). LECs exhibit complex error characteristics, which largely depend on the underlying training distribution. 
Deviations from this distribution, called \emph{distributional shifts}, arise from changes in operating conditions represented in the training dataset or environmental factors such as lighting variations and sensor occlusions \cite{robotics_ood_overview}. A significant distributional shift leads to degradation in measurement performance. Recent advances in sensing and perception strategies have addressed challenges related to lighting and occlusion \cite{low_lighting_jia, low_lighting_kim, occ_det_Moller, occ_det_Su}. However, measurement errors stemming from operating conditions, excluding factors such as lighting and occlusion, can be modeled as state-dependent errors~\cite{vision_based_swarm_control, sarah_dean_full_tehcnical_report}. In such scenarios, driving the system back to a region with low distributional shift necessitates preemptive control measures.

Traditional additive zero-mean gaussian noise models, which are widely used in the controls community, are not representative of such state-dependent measurement errors \cite{ERAS}. Despite this limitation, several safety-critical control techniques have been proposed to mitigate measurement uncertainty with LECs in the loop. In \cite{sarah_dean_full_tehcnical_report}, Dean et al. use a perception LEC for autonomous control of an Ackermann style robot. Their formulation is tightly coupled with an LQR controller, limiting its applicability to other control architectures. Additionally, they assume a maximum bound on the distribution shift of the LEC during deployment. A similar assumption was made to develop Measurement-Robust CBFs, which rely on a fixed global worst-case bound on measurement error to guarantee safety ~\cite{cbf_measurement_robust_dean,cbf_measurement_robust}. Although such techniques ensure safety under fixed worst-case bounds, they may result in overly conservative control in certain operating regions.

Belief Control Barrier Functions (BCBFs) \cite{belief_cbf} enforce probabilistic safety constraints without relying on fixed bounds on measurement errors. However, the original formulation assumes additive, zero-mean Gaussian noise that is independent of the state. In this work, we extend the BCBF framework to incorporate state-dependent measurement noise, which provides a more faithful representation of complex sensor modalities in safety-critical applications. Performance of BCBFs heavily depends on the accuracy of the underlying belief, which is propagated using an Extended Kalman Filter (EKF) algorithm. To maintain accurate belief updates under state-dependent measurement noise, we employ the Generalized EKF (GEKF) algorithm, which extends the EKF measurement update process to account for state-dependent measurement noise. Specifically, in this work, we:
\begin{itemize}
    \item Integrate the Generalized Extended Kalman Filter (GEKF) into the BCBF framework to enable safety-critical control under state-dependent measurement noise; and
    \item Demonstrate safer and less conservative control using the BCBF-GEKF approach through two simulated scenarios: (i) setpoint tracking with a 1D single-integrator system, and (ii) trajectory tracking with a 2D unicycle system.
\end{itemize}
The rest of the paper is organized as follows. In Sec.~\ref{sec:problem_setting}, we provide mathematical preliminaries and formally state our problem. Sec.~\ref{sec:technical_approach} describes our solution, where we pick a state-dependent noise model and extend the BCBF framework using GEKF. Sec.~\ref{sec:results} presents a comparative study of BCBF-GEKF versus BCBF-EKF through simulated results. Finally, we present our conclusion and outline future directions in Sec.~\ref{sec:conclusion}.


\section{Problem Formulation}
\label{sec:problem_setting}

We use the following notation throughout the paper. $\R$ and $\R_+$ denote the set of real and non-negative real numbers respectively. A multivariate gaussian variable $x$ is denoted as $x \sim \mathcal{N}(u, \Sigma)$, with mean vector $\mu$ and covariance $\Sigma$. A function $\phi: \R \rightarrow \R$ is an extended class-$\mathcal{K_{\infty}}$ function if $\phi(0)=0$ and $\phi$ is strictly increasing on the interval $(-\infty,\infty)$. The interior and boundary of a closed set $\mathcal{S}$ are denoted as $\mathrm{Int}(\mathcal{S})$ and $\partial \mathcal{S}$ respectively. The r-th order Lie derivative of a continuously differentiable function $V:\R^n \mapsto \R$ along a vector field $f: \R^n \mapsto \R^n$ at a point $x \in \R^n$ is denoted by $L^r_fV(x) \triangleq \frac{\partial^r V}{\partial x^r} f(x)$, where $r$ is omitted if $r=1$. Finally, $\mathbf{1}_n$ denotes a vector of ones of length $n$.

\subsection{Deterministic Forward Invariance}

Consider the following nonlinear, control-affine system:
\begin{equation}
    \dot{x} = f(x) + g(x)u,
    \label{eq:deterministic_system_model}
\end{equation}
where $x \in \mathbb R^n$ is the state, $u \in \mathcal{U} \subseteq \mathbb R^m$ is the control input, and functions $f: \mathbb R^n \mapsto \mathbb R^n$ and $g: \mathbb R^n \mapsto \mathbb R^{n \times m}$ are the known, locally Lipschitz drift vector and control matrix respectively. Let $h: \mathbb R^n \mapsto \mathbb R$ be a continuously differentiable function. A safe set $\mathcal{C} \subset \mathbb R^n$ is defined as the zero super-level set of $h$ such that
\begin{subequations}\label{eq:safe_set}
\begin{align}
    \mathcal{C} &= \{x \in \mathbb R^n \mid h(x) \geq 0\}, \\
    \partial \mathcal{C} &= \{x \in \mathbb R^n \mid h(x) = 0\}.
\end{align}
\end{subequations}

\begin{definition}[Forward Invariance]
    Set $\mathcal{C}$ is said to be \emph{forward invariant} with respect to system \eqref{eq:deterministic_system_model} if for every $x_0 \in \mathcal{C}$, $x(t) \in \mathcal{C} \;\forall t \geq 0$ where $x(0) = x_0$. 
\end{definition}

\begin{definition}[Control Barrier Function]
    Given a set $\mathcal{C}$ defined by \eqref{eq:safe_set} for a continuously differentiable function $h: \mathbb R^n \mapsto \R$ with $0$ a regular value, $h$ is a \emph{\acrfull{cbf}} for system \eqref{eq:deterministic_system_model} with respect to $\mathcal{C}$ if there exists a locally Lipschitz function $\phi \in \mathcal{K}_{\infty}$ such that, $\forall x \in \mathbb R^n$,
    \begin{equation}\label{eq:cbf_condition}
        \sup_{u \in \mathcal{U}}L_fh(x) + L_gh(x)u \geq -\phi\left(h(x)\right).
    \end{equation}
\end{definition}

\begin{theorem} [\hspace{-0.3pt}\cite{cbf_theory_and_applications}, Thm. 2]
    Given a set $\mathcal{C}$ as defined in \eqref{eq:safe_set} for a continuously differentiable function $h: \mathbb R^n \mapsto \mathbb R$, any locally Lipschitz continuous controller $u$ that satisfies \eqref{eq:cbf_condition} will render $\mathcal{C}$ forward invariant. 
    \label{thm:cbf_forward_invariance}
\end{theorem}

\subsection{Probabilistic Forward Invariance}
The guarantee of forward invariance stated in Theorem \ref{thm:cbf_forward_invariance} is based on knowledge of the true system state, which is hardly the case for real world systems. In such systems, we \emph{estimate} the state by taking into account process and measurement noise. Consider the following stochastic system model
\begin{subequations}\label{eq:stochastic_system_model}
    \begin{align}
    \dot{x} &= f(x) + g(x)u + w,\label{eq:stochastic_process_model} \\
    z &= \ell(x) + v, \label{eq:stochastic_measurement_model}
    \end{align}
\end{subequations}
where $f$ and $g$ are as defined for \eqref{eq:deterministic_system_model}, while $\ell: \mathbb R^n \mapsto \R^o$ is a non-linear observation function. This system is corrupted by independent process noise $w \sim \mathcal{N}(0, Q)$ for $Q$ $\in \R_+^{n \times n}$, and measurement noise $v \sim \mathcal{N}(\mu_v, R)$ for $\mu_v \in \R^o$ and $R \in \R_+^{o \times o}$. Since the forward invariance guarantee in theorem \ref{thm:cbf_forward_invariance} is defined for the true state $x$, it breaks down due to the presence of noise components because $x$ can no longer be accurately determined. Therefore, we need to modify CBF definition \eqref{eq:cbf_condition} to account for uncertainty.

To capture uncertainty, we need to determine the underlying distribution of the state and measurements. A Guassian distribution is widely used in robotics applications to represent stochastic system models \cite{optimal_and_robust_estimation, prob_robotics}. The unimodal assumption of a Guassian distribution holds true when a reliable initial estimate is available \cite{prob_robotics}. Additionally, if a good linear approximation of the system model \eqref{eq:stochastic_system_model} can be obtained, Gaussian state estimators offer tractable and accurate estimation \cite{prob_robotics}. For a multivariate Gaussian variable, the belief $b$ of state $x(t)$ in the stochastic process model \eqref{eq:stochastic_process_model} is represented as:
\begin{equation}\label{eq:belief}
    b = (\mu(t), \Sigma(t)),
\end{equation}
where $\mu \in \mathbb{R}^n$ and $\Sigma \in \mathbb{R}^{n \times n}$ represent the mean vector and covariance matrix respectively \cite{optimal_and_robust_estimation}. For non-linear drift vector $f$ and control matrix $g$ in process model \eqref{eq:stochastic_process_model}, the belief of the system can be propagated in continuous time using a Gaussian state estimator as \cite{optimal_and_robust_estimation}:
\begin{equation}
\aligned
\dot{\mu} &= f(\mu) + g(\mu)u, \\
\dot{\Sigma} &= F(\mu, u) \Sigma  + \Sigma F(\mu, u)^T + Q, 
\endaligned
\label{eq:state_estimator_time_update}
\end{equation}
where 
\begin{equation*}
    F(x, u) = \frac{\partial}{\partial x}\big(f(x) + g(x)u\big)
\end{equation*}
is the Jacobian of the deterministic part of the process model \eqref{eq:stochastic_process_model} with respect to the state vector, evaluated at the mean $\mu$. The Gaussian state estimator updates this belief under discrete measurements as:
\begin{equation}
\aligned
\mu^+ &= \mu^- + K_k (z_k - \hat{z_k}), \\
\Sigma^+ &= \Sigma^- - K_kH_k\Sigma^-, \\
\endaligned
\label{eq:state_estimator_measurement_update}
\end{equation}
where $z_k$ and $\hat{z_k}$ are the true and predicted measurements at the discrete time step $k$. $H_k = \left. \frac{\partial \ell}{\partial x} \right|_{x=\mu^{-}}$ is the Jacobian of the (non-linear) measurement model $\ell(x)$ with respect the state vector $x$, evaluated at the mean $\mu^-$. $K$ is an estimator-dependent linear gain matrix.

To enforce probabilistic forward invariance, safety constraints must be reformulated in the belief space. A common approach to express safety constraints for CBFs is to define linear inequalities with respect to the state, known as half-space constraints \cite{chance_constraint_path_planning, chance_constraint_mav, belief_cbf}. A half-space constraint is defined as $\alpha^T x - \beta \geq 0$, where $\alpha \in \mathbb{R}^n$ is a coefficient vector and $\beta \in \mathbb{R}$ is a scalar offset. To map this half-space constraint to belief-space, we compute the probability of satisfying the half space constraint for a given belief $b = (\mu, \Sigma)$ of the state $x$ as a chance constraint \cite{chance_constraint_path_planning}:
    \begin{equation}
        \Pr[\alpha^T x - \beta \geq 0] = \frac{1}{2} \left( 1 + \text{erf} \left( \frac{\alpha^T \mu - \beta}{\sqrt{2 \alpha^T \Sigma \alpha}} \right) \right),
    \end{equation}
where $\text{erf}$ is the standard error function \cite{chance_constraint_path_planning}. To capture the expected worst-case behavior at the tail-end of the distribution of $\alpha^T x - \beta$, we compute the expected cost of violating this constraint for a given risk level $\delta$ using Conditional Value at Risk (CVaR) - a coherent risk metric \cite{cvar}. $\text{CVaR}_{\delta}(\alpha^Tx - \beta)$ for given belief $b = (\mu, \Sigma)$ is computed as \cite{CVaR_BPOE}: 
\begin{equation}
    \text{CVaR}_{\delta}(\alpha^Tx - \beta) = \alpha^T \mu \;-\; \sqrt{\alpha^T \Sigma \alpha}\; \frac{f\!\left(q_{\delta}(\frac{\alpha^T x - \beta}{\sqrt{\alpha^T \Sigma \alpha}})\right)}{\delta},
    \label{eq:CVaR}
\end{equation}
where $q_\delta(x)$ and $f$ are the $\delta$-quantile and probability density function of a standard normal gaussian variable. For a given half-space constraint $\alpha^T x - \beta \geq 0$ and risk level $\delta$, the following are equivalent \cite{belief_cbf}:
\begin{equation}
    \Pr\!\big(\alpha^Tx - \beta  \geq 0 \big) \geq 1 - \delta \;\;\Longleftrightarrow\;\; 
    \text{CVaR}_{\delta}\!\big(\alpha^T x - \beta \big) \geq 0, \label{eq:CVaR_constraint}
\end{equation}
We can now construct a \emph{belief safe set} $\mathcal{C}_b$ defined as:
\begin{equation}
    \mathcal{C}_b = \{b \in \mathbb{R}^{n_b} \mid h_b(b) \geq 0 \},
    \label{eq:belief_safe_set}
\end{equation}
where $n_b = \frac{n^2 + 3n}{2}$ is the dimension of the belief space \cite{belief_cbf} where $h_b(b)$ is a \emph{belief barrier function} defined as:
\begin{equation}
\begin{aligned}
    h_b(b) &:= \text{CVaR}_{\delta} \!\left( \alpha^T x - \beta \right) \\
    &= \alpha^T \mu - \beta \;-\; \sqrt{\alpha^T \Sigma \alpha}\;
       \frac{f\!\big(q_{\delta}\big)}{\delta}.
\end{aligned}
\label{eq:belief_cbf}
\end{equation}
\subsection{State-dependent measurement noise}
Complex sensor modalities like vision-based LECs do not exhibit ``fixed" noise characteristics: the change in measurement error with respect to the state is monotonic and heteroskedastic \cite{ERAS}. While \eqref{eq:stochastic_measurement_model} models stochastic measurement noise, the noise itself is independent of the state. Therefore, it does not represent an accurate sensor model for such modalities. To complete our problem formulation, we now consider a system with the following process and measurement model:
\begin{subequations}\label{eq:state_dependent_stochastic_system_model}
    \begin{align}
        \dot{x} &= f(x) + g(x)u + w, \\
        z &= \ell(x) + \aleph(x) \label{eq:generic_state_dependent_measurement_model}
    \end{align}
\end{subequations}
where $\aleph(x)$ is a state-dependent noise term, independent of the process noise $w$. We now define our main problem.

\begin{problem}\label{pb:main} Given a risk level $\delta \in (0, 1]$ and a belief safe set $\mathcal{C}_b$ \eqref{eq:belief_safe_set}, find a controller $u$ that satisfies $\Pr[\alpha^T x - \beta \geq 0] \geq 1 - \delta \; \forall \; t \geq 0$ and $h_b(b) \geq 0$ for the control affine system \eqref{eq:state_dependent_stochastic_system_model} using a state estimator of the form \eqref{eq:state_estimator_time_update}-\eqref{eq:state_estimator_measurement_update}. 
\end{problem}

\section{Technical Approach}\label{sec:technical_approach}
\subsection{Measurement and state estimation model}\label{sec:gekf}
To solve Problem \ref{pb:main}, we need to ensure that our estimator of choice maintains an accurate belief of the system \eqref{eq:state_dependent_stochastic_system_model}. Choice of the measurement noise model \eqref{eq:generic_state_dependent_measurement_model} affects convergence of the belief through the estimator's measurement update step \eqref{eq:state_estimator_measurement_update}. For vision-based LECs, state-dependent worst-case measurement errors have been approximated as linear \cite{vision_based_swarm_control} and polynomial \cite{sarah_dean_full_tehcnical_report} functions of the state. In this work, we model the worst-case measurement error as state-dependent noise that is linear in the state, defined by the following measurement model:
\begin{equation}\label{eq:mult_noise_equation}
    \aleph(x) = p \ell(x) + v,
\end{equation}
where $p \sim \mathcal{N}(\mu_p\mathbf{1}_o, P)$ with $\mu_p\ \in \R$ and $P \in \R_+^{o \times o}$ is the \emph{multiplicative} component which enforces state-dependence of the measurement noise in a linear fashion. The non-linear observation function $\ell(x)$ and additive noise $v$ are the same as defined for \eqref{eq:stochastic_measurement_model}. Note that $p$ is independent of $v$. Although a linear model might be stringent for the polynomial worst-case assumption in \cite{sarah_dean_full_tehcnical_report}, it is less conservative than the widely adopted alternatives of an additive zero-mean gaussian noise model or fixed global measurement error/uncertainty bound assumptions used previously \cite{FxTS_RaCBF, RCBF_uncertainty_bound, agrawal2022safe}. This type of linear noise model has been employed for state estimation and optimal control of a sensorimotor system \cite{mult_noise_sensorimotor}, and target tracking in a Wireless Network System using range-baring sensors \cite{mult_noise_WSN}. Based on the choice of this noise model, our final system model becomes:
\begin{subequations}\label{eq:mult_stochastic_system_model}
    \begin{align}
    \dot{x} &= f(x) + g(x)u + w\\
    z &= (1 + p) \ell(x) + v \label{eq:state_dependent_measurement_model}
    \end{align}
\end{subequations}
To obtain an accurate belief, we need to account for this noise in the measurement update step of form \eqref{eq:state_estimator_measurement_update}. In \cite{GEKF}, authors propose the Generalized Extended Kalman Filter (GEKF) algorithm, which builds on  \eqref{eq:state_estimator_measurement_update} by propagating multiplicative noise $p$ explicitly through the covariance update during a measurement:
\begin{equation}\label{eq:gekf_covariance_update}
    \Sigma^+  = \Sigma^- - (1 + \mu_p) K_{G_k} H_k \Sigma^-,
\end{equation}
where $K_G$ is a linear (Kalman) gain, and $H_k$ is as defined for (\ref{eq:state_estimator_measurement_update}). $K_G$ also implicitly captures the propagation of multiplicative noise $p$ and additive noise $v$ through the innovation covariance $S_{G_k}$:
\begin{equation}
     K_{G_k}   = (1 + \mu_p) H_k \Sigma^- S_{G_k}^{-1}.
\end{equation}
 $S_{G_k}$ is defined for multiplicative noise covariance $P = \sigma_p^2I_o, \sigma_p \in \mathbb{R}$ as:
\begin{equation}
    S_{G_k} = (1 + \mu_p)^2 H_k \Sigma^- H_k^T + \sigma_p^2 M_k + R,\label{eq:GEKF_S_def}\\
\end{equation} where:

\begin{equation*}
    M_k = \mathrm{diag} \left\{
    H_k \Sigma^- H_k^T + \ell(\mu^-)\ell(\mu^-)^T \right\}.
\end{equation*}
Here, the term $H_k \Sigma^- H_k^T$ projects the belief covariance into the 
observation space $\mathbb{R}^o$, scaled by the multiplicative noise mean factor 
$(1+\mu_p)^2$. The second term, $\sigma_p^2 M_k$, accounts for variance 
introduced by state-dependent multiplicative noise; specifically, the term 
$\sigma_p^2 \ell(\mu^-)\ell(\mu^-)^T$ reflects how uncertainty grows with the 
expected measurement $\ell(\mu^-)$. $R$ captures the contribution of the 
additive measurement noise $v$. Finally, $p$ and $v$ are propagated to the mean $\mu$ of the belief $b$ through $K_{G_k}$ as:
\begin{equation}
    \mu^+ = \mu^- + K_{G_k}\left[z_k - \hat{z_k}\right],
\end{equation}
where $z_k$ is the actual noisy measurement $\hat{z_k}$ is the predicted noisy measurement value at time step $k$. $\hat{z_k}$ is approximated by taking the expected value of the first order Taylor series approximation of $z_k$, evaluated at the mean $\mu^-$ of the prior belief $b^-$ \cite{GEKF}:
\begin{equation}
    \hat{z} \approx \mathbb{E}\left[z\right] = (1 + \mu_p) \ell(\mu^-) + \mu_v \mathbf{1}_o.
    \label{eq:GEKF_observation}
\end{equation}
We complete the solution to Problem ~\ref{pb:main} by employing Belief-Control Barrier Functions \cite{belief_cbf} in conjunction with the GEKF measurement update. We restate the BCBF definition here for reference:

\begin{definition}[BCBF~\cite{belief_cbf}]
Given a safe set $\mathcal{C}_b$, $h_b(b)$ is defined as a \textit{Belief Control Barrier Function (BCBF)} for the stochastic dynamical system \eqref{eq:state_dependent_stochastic_system_model} if $\forall b$ satisfying $h_b(b) \geq 0, \ \exists u \in \mathcal{U}$ such that
\begin{equation}
    \frac{\partial h_b}{\partial b} \left( f_b(b) + g_b(b) u \right) \geq -h_b(b).
    \label{eq:BCBF}
\end{equation}
\end{definition}
In the following section, we discuss how the innovation covariance $S_{G_k}$ impacts probabilistic guarantees of leaving or staying in the safe set during a discrete measurement update, which were previously introduced in \cite{belief_cbf} for a fixed-noise measurement model.

\subsection{Probabilistic Safety Guarantees under multiplicative state-dependent noise}
In \cite{belief_cbf}, Vahs et al. use an Extended Kalman Filter with continuous-time prediction \eqref{eq:state_estimator_time_update} and discrete-time update steps which propagates the belief of the system \eqref{eq:stochastic_system_model}. The discrete update steps can potentially cause the belief $b$ to discretely jump outside the safe set safe set $\mathcal{C}_b$. The bounds on the probabilities of the belief leaving or staying in the belief safe-set depend on the covariance of the zero-mean innovation term $\theta = K_k \left( z_k - \ell(\mu^{-}) \right)$, where $z_k$ is the observation obtained from the sensor and $K_k$ is the Kalman gain of an Extended Kalman Filter (EKF) state estimator at discrete time step $k$. This covariance is defined in \cite{belief_cbf} as:
\begin{equation}
    \Lambda = K_k S_k K_k^{T},
    \label{eq:innovation_covariance}
\end{equation}
where $S_k$ is the innovation covariance of the EKF. $K_k$ and $S_k$ are defined as:
\begin{subequations}\label{eq:ekf_update}
\begin{align}
    K_k &= \Sigma^- H_k^T S_k^{-1}. \\
    S_k &= H_k \Sigma^- H_k^T + R
\end{align}
\end{subequations}
where $H_k$ is the jacobian of the non-linear observation function as defined for (\ref{eq:state_estimator_measurement_update}) and $R$ is the covariance of the additive zero-mean gaussian noise in \eqref{eq:stochastic_measurement_model}. Theorems 2 and 3 in \cite{belief_cbf} state the probability of the belief leaving or staying in the safe set, respectively. These are restated as Theorems \ref{th:Theorem1} and \ref{th:Theorem2} using the CVaR definition \eqref{eq:CVaR} below:

\begin{theorem} [\hspace{-0.3pt}\cite{belief_cbf}, Thm. 2]
If the control input \( u(t) \) satisfies \eqref{eq:BCBF}, the probability of leaving the safe set under a discrete transition, i.e., $\Pr\left[h_b(b^+) < 0\right]$, is bounded by
\begin{equation}
    \Pr\left[h_b(b^+) < 0\right] \leq \frac{1}{2} \left( 1 - \operatorname{erf} \left( 
    \frac{\xi(b^-)}{\sqrt{2 \alpha^T \Lambda \alpha}} \right) \right),
\end{equation}
\label{th:Theorem1}
where
\begin{equation*}
    \xi(b^-) \;=\; \frac{f(q_{\delta})}{\delta} 
    \left( \sqrt{2 \alpha^T \Sigma^{-} \alpha} 
    - \sqrt{2 \alpha^T (I - K_k H_k)\, \Sigma^{-} \alpha} \right).
\end{equation*}
\end{theorem}

\begin{theorem}[\hspace{-0.3pt}\cite{belief_cbf}, Thm. 3]
    The belief state \( b^+ \) remains in the safe set \( \mathcal{C}_b \)
    with probability \( \Pr\left[b^+ \in \mathcal{C}_b\right] \geq 1 - \varepsilon \)
    during the discrete belief update from $b^-$ to $b^+$ if the control input \( u(t) \) satisfies equation \eqref{eq:BCBF} for an augmented BCBF function \( \tilde{h}_b(b) = h(b) - \gamma \) and \( \tilde{h}_b(b^-) \geq 0 \) for
    \begin{equation*}
        \gamma \geq \sqrt{2 \alpha^T \Lambda \alpha} 
        \left( \operatorname{erf}^{-1}(1 - 2\varepsilon) \right) - \xi(b^-),
    \end{equation*}
    where $\varepsilon$ is a user-defined risk value.
    \label{th:Theorem2}
\end{theorem}

Here $b^-$ and $b^+$ represent the prior and posterior beliefs computed before and after the discrete measurement update respectively. Note that $S_{G_k}$ can be written as:

\begin{equation*}
    S_{G_k} \;=\; S_k \;+\; \mu_p^2\,H_k \Sigma_k^- H_k^T \;+\; \sigma_p^2\, M_k,
\end{equation*}
where $M_k \succeq 0$. Therefore, $S_{G_k} \succeq S_k$, given $\mu_p \ge 0$ and $\sigma_p^2>0$. Since
\begin{equation*}
\begin{aligned}
    K_{G_k} \;&=\; (1+\mu_p)\,\Sigma_k^- H_k^T S_{G_k}^{-1}, \text{ and} \\
    K_k \;&=\; \Sigma_k^- H_k^T S_k^{-1},
\end{aligned}
\end{equation*}
a larger innovation covariance yields a smaller Kalman gain (in norm). Therefore, GEKF places less weight on its measurements, leading to more robust  estimation under state-dependent measurement noise. Since the covariance of the innovation term $\Lambda$, defined in \eqref{eq:innovation_covariance}, depends on the estimator’s innovation covariance, it directly influences the probability bounds in Theorems~\ref{th:Theorem1} and~\ref{th:Theorem2} for leaving or remaining in the safe set $\mathcal{C}_b$. Since $\Lambda$ only holds for observation model \eqref{eq:stochastic_measurement_model}, the probabilistic bounds guaranteed by these theorems would not be valid under the assumption of the state-dependent measurement model \eqref{eq:state_dependent_measurement_model}. Therefore, we need to recompute the mean and covariance of the innovation term to represent more accurate probabilistic bounds for BCBFs operating under observation model \eqref{eq:state_dependent_measurement_model}.


\subsection{Innovation covariance under a multiplicative noise model}\label{sec:innovation}
For a given measurement realization $z_k$ and its expectation $\hat{z}_k$, the innovation term for the GEKF is defined as
\begin{equation}
   \theta_G = K_G(z_k - \hat{z}_k)\:.
   \label{eq:GEKF_innovation_term}
\end{equation}
To derive the expression for the mean of the innovation term, a Taylor series approximation can be used to calculate the propagation of mean and covariance of a random variable $x$ through a non-linear function $\eta$ as~\cite{belief_cbf}:
\begin{align}
\mathbb{E}\{\eta(x)\} &\approx \eta(\mathbb{E}\{x\}), \label{eq:assumption_1_mean} \\
\text{Var}\{\eta(x)\} &\approx \left( \frac{\partial \eta}{\partial x} \mathbb{E}\{x\} \right) \text{Var}\{x\} \left( \frac{\partial \eta}{\partial x} \mathbb{E}\{x\} \right)^T \label{eq:assumption_1_cov}\:.
\end{align}
\begin{proposition}\label{prop:zero_mean}
    The innovation term~\eqref{eq:GEKF_innovation_term} is a zero-mean random variable.
\end{proposition}
\begin{proof}
The mean of the GEKF innovation term is obtained from its expected value as:  
\begin{equation*}
    \mathbb{E}[\theta_G] = K_G \big( \mathbb{E}[z_k] - \mathbb{E}[\hat{z}_k] \big).
\end{equation*}
From \eqref{eq:GEKF_observation}, the expectation of the predicted measurement is evaluated as:
\begin{equation*}
    \mathbb{E}[\hat{z}_k] = (1 + \mu_p)\,\ell(\mu^-) + \mu_v \mathbf{1}_o,
\end{equation*}
while the expectation of the true measurement is given by:
\begin{equation*}
    \mathbb{E}[z_k] = (1 + \mu_p)\,\mathbb{E}[\ell(x)] + \mu_v \mathbf{1}_o.
\end{equation*}
Using \eqref{eq:assumption_1_mean}, the expectation of the observation function becomes:
\begin{equation*}
    \mathbb{E}[\ell(x)] = \ell(\mathbb{E}[x]).
\end{equation*}
Since $\mathbb{E}[x] = \mu^-$, it follows that:
\begin{equation*}
    \mathbb{E}[z_k] = (1 + \mu_p)\,\ell(\mu^-) + \mu_v \mathbf{1}_o.
\end{equation*}
Thus, the expected measurement and predicted measurement coincide, and the 
innovation term of the GEKF is a zero-mean random variable:
\begin{equation}
    \mathbb{E}[\theta_G] 
    = K_G \big( \mathbb{E}[z_k] - \mathbb{E}[\hat{z}_k] \big) 
    = 0.
    \label{eq:GEKF_innovation_term_mean}
\end{equation}
\end{proof}
\begin{proposition}\label{prop:innovation_covariance}
The covariance for the GEKF innovation term is given by:
\begin{equation}
    \Lambda_G = K_{G_k} S_{G_k}K_{G_k}^T,
\end{equation}
where $K_{G_k}$ and $S_k$ are the GEKF Kalman gain and innovation covariance respectively \cite{GEKF}.
\end{proposition}
\begin{proof}
Since the innovation term $\theta_G$ is defined as the product of the linear Kalman gain $K_{G_k}$ with the innovation $(z_k - \hat{z_k})$, the covariance of the innovation term is obtained by propagation of the Kalman gain through the innovation covariance $S_{G_k}$.
\end{proof}

\section{Experimental Results}
\label{sec:results}
To demonstrate the efficacy of our proposed approach, we employ a BCBF with the GEKF under noise model \eqref{eq:state_dependent_measurement_model} in two simulated scenarios: (i) a nonlinear single-integrator system tasked with reaching a desired setpoint while avoiding a half-space safety constraint, and (ii) a 2D unicycle kinematic system tracking a sinusoidal trajectory, with safety constraints imposed at the trajectory extremities. We use the BCBF-EKF formulation in \cite{belief_cbf} as a baseline for our results. In both scenarios, the time update \eqref{eq:state_estimator_time_update} is implemented at 1000 Hz. We use a ZoH framework \cite{zoh_signal_processing} to obtain measurements at 0.1 Hz. This exaggerates the effects of discrete, sporadic updates to the belief of the system during the measurements, which allows us to highlight the performance of our approach against the BCBF-EKF baseline. 

\subsection{Setpoint tracking for a single integrator system}
Consider a one-dimensional nonlinear single-integrator system with dynamics and measurement model given by
\begin{align*}
\dot{x} &= 0.1 \cos(x) + u, \\
z &= (1 + p)\ell(x) + v,
\end{align*}
where $\ell(x) = Ix$ is the identity observation function. The multiplicative noise $p$ and additive noise $v$ are characterized by $\mu_p = 0.1, P = 0.001^2I_1$ and $\mu_v = 0.01, R= 0.0005^2I_1$, respectively. This allows the multiplicative noise $p$ to dominate, facilitating the demonstration of how it is handled by GEKF vs EKF. The system is driven to a desired state of \( x = 6.0 \) by a Control Lyapunov Function (CLF) based nominal control law $V(x) = (x - 6.0)^2$.  We constrain the belief of the system to $\mathcal{C}_b = \left\{ b \in \mathbb{R} \,\middle|\, \alpha^T x \geq \beta \right\}$, where $\alpha=-1$ and $\beta=-5.0$. This belief safe set represents the half space constraint $x \leq 5.0$. The optimal control input that obeys this CLF law while observing a BCBF constraint is determined using the following CLF-CBF-QP:
\begin{align*}
u(x) = \argmin_{(u, \rho) \in \mathbb{R}^{m+1}} 
\quad & \frac{1}{2} u^T H u + s \rho^2  \\
\text{s.t.} \quad 
& L_f V(x) + L_g V(x) u \leq -V(x) + \rho, \\
& L_f h_b(b) + L_g h_b(b) u \geq -h_b(b),\\
& -u_{max} \leq u \leq u_{max},
\end{align*}
where
\begin{equation*}
H=
\begin{bmatrix}
1 & 0 \\
0 & 2s
\end{bmatrix},
\end{equation*}
$s = 10.0$ is the slack penalty, and $\rho$ is the slack variable that is optimized to relax the CLF requirement for prioritizing safety. The maximum absolute control value is set to $u_{max}=1$. We set the desired CVaR risk level to \( \delta = 0.001 \).

Table \ref{tab:gekf_ekf_metrics} summarizes the performance comparison between the BCBF-GEKF and BCBF-EKF controller over the course of 100 Monte Carlo simulation runs. Due to higher tracking accuracy, the BCBF-GEKF controller is less conservative as compared to the BCBF-EKF. The standard EKF fails to account for the bias induced from the multiplicative and additive noise. As a result, it fails to converge to the true trajectory of the system and consistently overestimates the true position, leading to repeated violations near the safety boundary with respect to the estimated trajectory (Fig.~\ref{fig:1d_traj}).

\begin{figure}
    \centering
    \includegraphics[width=1\linewidth]{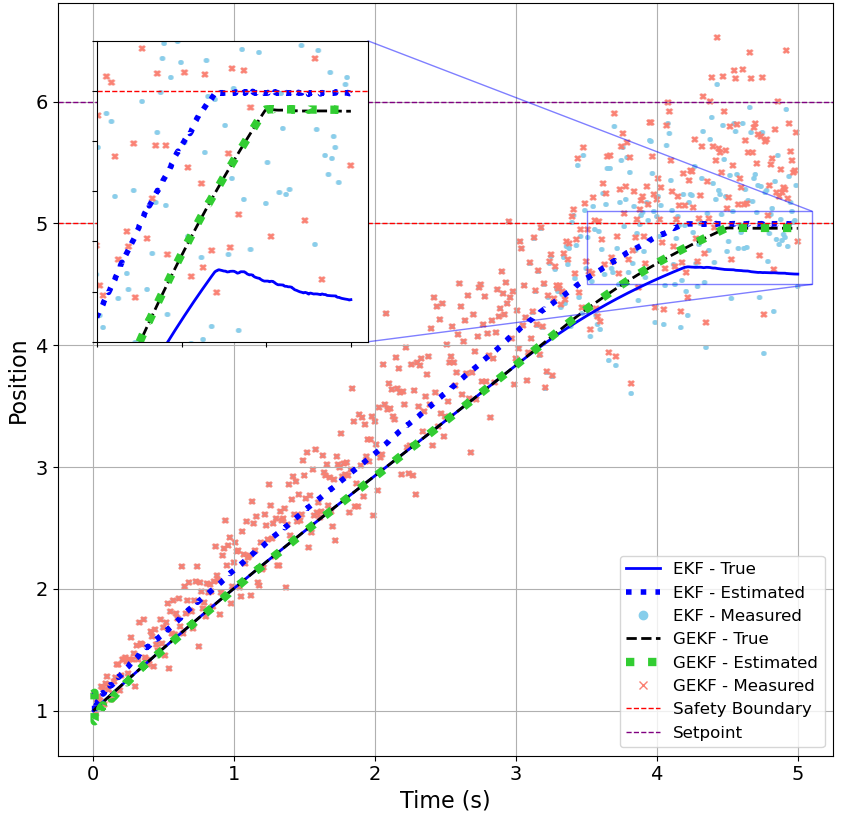}
    \caption{Estimated and true trajectories for the 1D nonlinear single-integrator setpoint tracking task under BCBF-GEKF and BCBF-EKF control. The inset highlights that the EKF estimate crosses the safety boundary due to large estimation error, even though the true system trajectory remains safe. In contrast, the GEKF estimate respects the safety boundary while enabling the true trajectory to operate closer to it.}
    \label{fig:1d_traj}
\end{figure}
Conversely, BCBF-GEKF achieves a slightly lower estimated distance from the safety boundary while allowing the actual trajectory of the system to get closer to the desired setpoint. This characteristic of the BCBF-GEKF can be exploited to allow a system with multiplicative state-dependent noise to navigate more constrained environments.





\begin{table} 
    \centering
    \begin{tabular}{l c c}
        \hline
        \textbf{Metric} & \textbf{BCBF-GEKF} & \textbf{BCBF-EKF} \\
        \hline
        \% estimated exceedances per run & 0.000 & 3.307 \\
        \% true exceedances per run & 0.000 & 0.000 \\
        Max estimated position value & 4.958 & 5.004 \\
        Max true position value & 4.958 & 4.636 \\
        Mean est distance from boundary & 1.693 & 1.524 \\
        Average controller effort & 62.573 & 60.222 \\
        Tracking RMSE & 0.013 & 0.248 \\
        \hline
    \end{tabular}
    \caption{Comparison of BCBF-GEKF and BCBF-EKF performance metrics averaged over 100 simulations.}
    \label{tab:gekf_ekf_metrics}
\end{table}

\subsection{2D Sinusoidal Trajectory Tracking}
\subsubsection{System Model}
Consider a unicycle vehicle model $x = [x,\,y,\,v,\,\theta]^T$ subject to zero-mean Gaussian process noise  $w \sim \mathcal{N}(0, Q)$, where $(x,y)$ denotes position in the plane, $v$ is the longitudinal velocity, and $\theta$ is the heading angle. The system is affine in control input $u$ such that
\begin{equation}
    \dot{x} = f(x) + g(x)u + w,
    \label{eq:dubins_affine}
\end{equation}
where
\[
f(x) =
\begin{bmatrix}
v \cos \theta \\
v \sin \theta \\
0 \\
0
\end{bmatrix},
\qquad
g(x) =
\begin{bmatrix}
0 & 0 \\
0 & 0 \\
1 & 0 \\
0 & 1
\end{bmatrix},
\qquad
u =
\begin{bmatrix}
a \\ \omega
\end{bmatrix}.
\]

Here $a$ and $\omega$ are the longitudinal acceleration and yaw-rate respectively. The observation model is given by
\begin{equation}
    z(x) = \bigl[1 + p_n\bigr]\,\ell_n + v_n.
    \label{eq:obs_model}
\end{equation}
where $\ell$ is the observation function that maps the state to its y component ($\ell(x, y, v, \theta) = y$). We set $Q = 0.0001I_4$, $\mu_p = 0.1$, $P = 0.01^2I_1$, $\mu_v = 0.001$ and $R = 0.0005^2I_1$.

\subsubsection{Target trajectory}
The target trajectory $x_d$ is generated based on a sinusoidal function parameterized by time

\begin{equation}
x_d\;=\;
\begin{bmatrix}
x_d(t) \\[4pt]
y_d(t) \\[4pt]
v_d(t) \\[4pt]
\theta_d(t)
\end{bmatrix}
=
\begin{bmatrix}
v_rt \\[4pt]
A \sin(\omega t) + A \\[4pt]
v_r \\[4pt]
\arctan\!\big(\tfrac{y_d}{x_d}\big)
\end{bmatrix},
\end{equation}
where $A=1.0$ is the amplitude of the wave, $\omega=0.5$ is the phase, and $v_r=1.0$ is the reference longitudinal velocity. 

\subsubsection{Gain-scheduled trajectory tracking controller}

The nominal control input $u_{nom}$ is generated by a gain-scheduled feedback controller \cite{gain_scheduling} obtained by linearizing the unicycle vehicle model around a straight-line trajectory with reference 
speed $v_r$ and zero heading. The state vector is defined as 
$x = [x,\,y,\,v,\,\theta]^T$, while the desired trajectory is 
$x_d = [x_d,\,y_d,\,v_d,\,\theta_d]^T$. The error coordinates are 
$e_x = x - x_d$, $e_v = v - v_r$, $e_y = y - y_d$, and 
$e_\theta = \theta - \theta_d$. The controller outputs longitudinal 
acceleration $a$ and the yaw rate $\omega$ according to 
\[
u_{nom} = u_d - K e,
\]
with the desired control input $u_d = [0,\,0]^T$. The feedback gain matrix is 
\[
K = \begin{bmatrix}
k_x & k_v & 0 & 0 \\
0 & 0 & k_y & k_\theta
\end{bmatrix},
\]
where the gains are scheduled as $k_x = \lambda_1$, $k_y = \tfrac{a_1}{v_r}$, 
and $k_\theta = a_2$. Here $\lambda_1=1.0,\,k_v=1.0,\,a_1=16.0$ and $\,a_2=100.0$ are design parameters 
that determine the closed-loop error dynamics. This results in a nominal 
controller of the form
\[
u_{\text{nom}} =
\begin{bmatrix}
    a \\ 
    \omega
\end{bmatrix}
=
\begin{bmatrix}
    -k_x e_x - k_v e_v \\ 
    -k_y e_y - k_\theta e_\theta
\end{bmatrix},
\]
where $a$ denotes the commanded longitudinal acceleration and $\omega$ the commanded yaw rate.  

\subsubsection{Safe Control}
To enforce safety while tracking the nominal trajectory, we impose the following half-space 
constraints at the peaks of the sinusoidal wave:
\begin{align*}
    \Pr\!\left[ \alpha_i^T x\;\geq\; \beta_i \right] &\;\geq\; 1 - \delta,
    \quad i = 1,2, \\[6pt]
    \Longleftrightarrow\;\; &\text{CVaR}_{\delta}\!\big(\alpha^T x - \beta \big) \geq 0 \implies h_{b_i} \geq 0\\
     \alpha_1^T &= \begin{bmatrix} 0 & -1 & 0 & 0 \end{bmatrix}, 
    \beta_1 = -5, \\[6pt]
    \alpha_2^T &= \begin{bmatrix} 0 & 1 & 0 & 0 \end{bmatrix}, 
    \beta_2 = -5.
\end{align*}
Here $h_{b_1}$ and $h_{b_2}$ enforce safety at the top and bottom peaks of the sinusoidal wave respectively (Fig. \ref{fig:2d_traj}). We set the CVaR risk level to $\delta=0.001$. The half-space constraints have a relative degree of 2 with respect to the unicycle model \cite{ECBFs}. Therefore, safety is enforced by respecting the second order constraints of the form:
\begin{align}
    L_g L_f h_{b_i}\, u \;\geq\;
    &  k\cdot\begin{bmatrix} \zeta_1 & \zeta_2 \end{bmatrix} \begin{bmatrix} L_f h_{b_i} \\[4pt] h_{b_i}(x) \end{bmatrix} - L_f^2 h_{b_i}
    \label{eq:cbf_rel_constraint}
\end{align}
where $\zeta_1=1.0$ and $\zeta_2=0.75$. $k=50.0$ is a linear gain which weights the safety constraints in an optimization process. This allows us to solve for the optimal control input that tracks the nominal trajectory while respecting the safety constraints:
\begin{equation}
\label{eq:qp}
\begin{aligned}
    u_{opt}(x) &= \argmin_{(u) \in \mathbb{R}^{m+1}} 
    \quad \frac{1}{2} (u - u_{nom})^T(u - u_{nom})\\
    \text{s.t.} \quad 
        & \text{\eqref{eq:cbf_rel_constraint} holds for $h_{b_1}, h_{b_2}$}, \\
        & -u_{\max} \leq u \leq u_{\max}.
\end{aligned}
\end{equation}
The system is initialized at $x_0 = [0, 0, 5.0, 0.45]$, with the estimator belief given by a noisy measurement of $x_0$ (see~\eqref{eq:state_dependent_measurement_model}) and covariance $\Sigma_0 = \operatorname{diag}([0.1, 0.1, 0.1, 0.1])$. Control inputs are bounded by $u_{\max} = [1, 1]^T$, and the quadratic program~\eqref{eq:qp} is solved using the JAXOPT BoxOSQP solver~\cite{jax, jaxopt}. Table~\ref{tab:gekf_ekf_comparison} reports mean metrics over 100 Monte Carlo runs of the 2D sinusoidal tracking task.  

The GEKF, by accounting for multiplicative noise, achieves an order-of-magnitude improvement in tracking accuracy over the EKF, while maintaining zero safety violations. Similar to the 1D case, GEKF allows the true trajectory to track the nominal sinusoid more closely, yielding a less conservative controller. In contrast, the EKF shows lower average and final covariance traces due to overconfidence arising from a larger Kalman gain that favors measurements more strongly. In addition to violating the top safety constraint (BCBF~1), about $14\%$ of BCBF-EKF runs produced failed trajectories, Failures typically occurred due to large estimation errors in $x$ (initially) or both $x$ and $y$ (after the noisiest $y$-measurement). In all failed cases, the system was unable to recover and reenter the safety set under BCBF-EKF control.

\begin{figure} 
    \centering
    \includegraphics[width=1\linewidth]{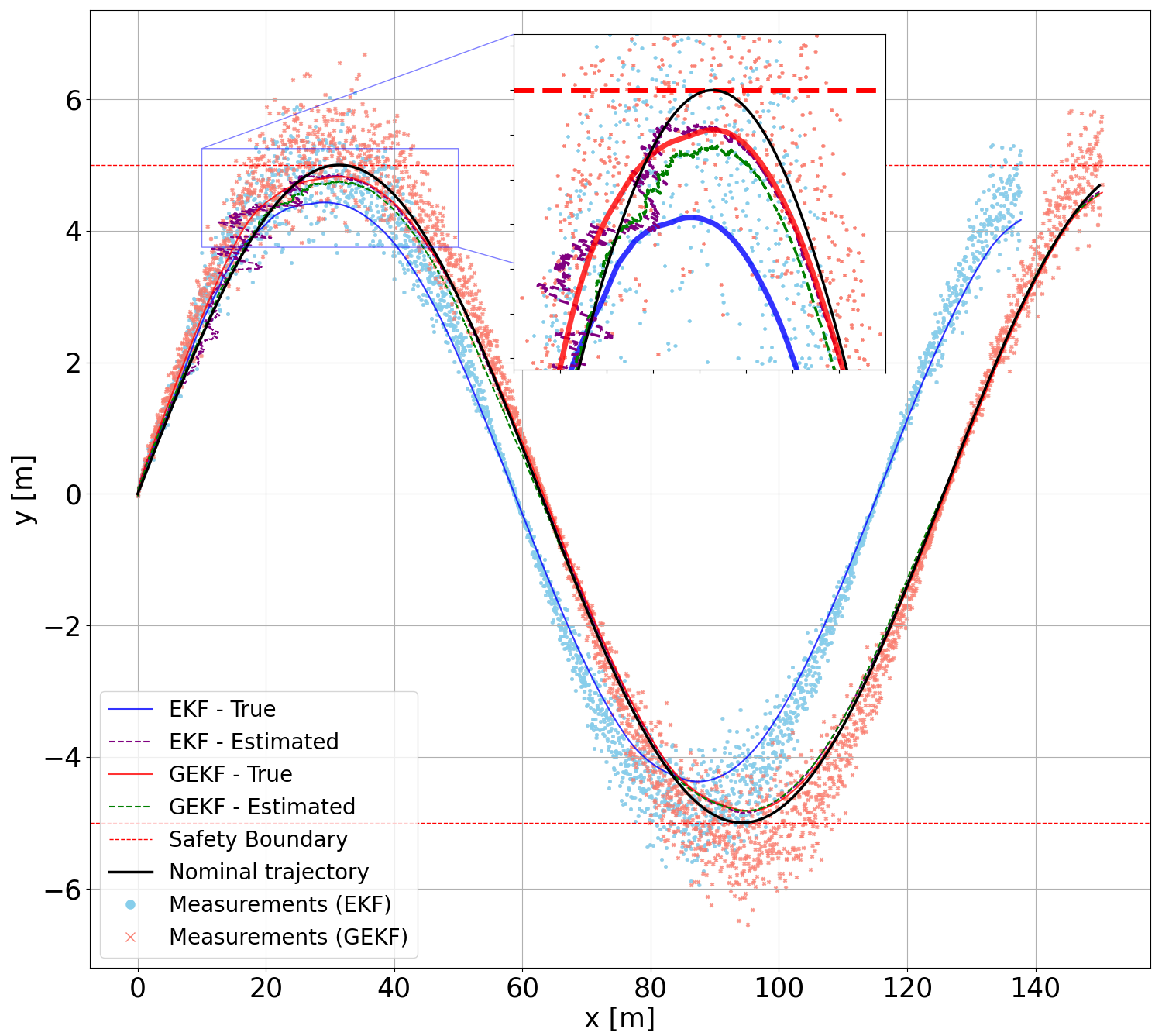}
    \caption{Comparison of trajectories, measurements, and estimates for sinusoidal tracking with BCBF-GEKF vs. BCBF-EKF. GEKF achieves closer nominal tracking while preserving safety comparable to the EKF.}
    \label{fig:2d_traj}
\end{figure}

\begin{figure} 
    \centering
    \includegraphics[width=1\linewidth]{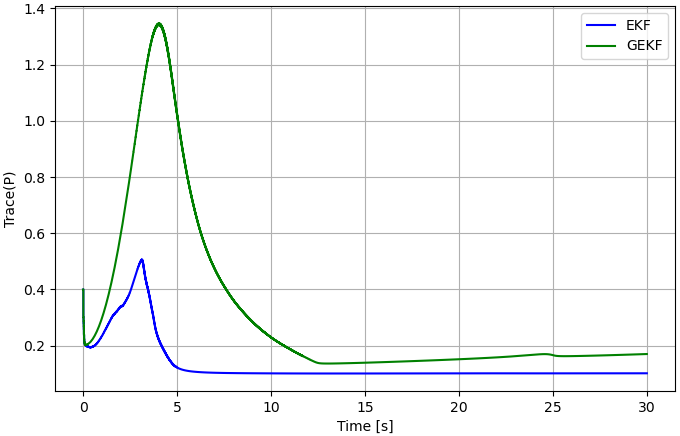}
    \caption{Evolution of trace of covariance under EKF and GEKF. Lower covariance trace of the EKF can be attributed to overconfidence in measurements due to a higher Kalman gain as compared to the GEKF.}
    \label{fig:2d_cov}
\end{figure}

\begin{table} 
\centering
\caption{Comparison of performance metrics across GEKF and EKF.}
\begin{tabular}{l|c|c}
\hline
\textbf{Metric} & \textbf{GEKF} & \textbf{EKF} \\
\hline
Max true $y$ value              & 4.757 & 4.438 \\
Min true $y$ value              & -4.819 & -4.388 \\
Max estimated $y$ value         & 4.757 & 4.897 \\
Min estimated $y$ value         & -4.823 & -4.846 \\
Tracking RMSE                   & 0.341 & 3.573 \\
Average acceleration            & 0.001 & -0.015 \\
Average yaw rate                & -0.013 & -0.013 \\
\% of BCBF 1 violations per run     & 0.000 & 0.025 \\
\% of BCBF 2 violations per run  & 0.000 & 0.000 \\
Mean covariance trace           & 0.301 & 0.128 \\
Final covariance trace          & 0.165 & 0.100 \\
Failure rate (\%)               & 0.000 & 13.821 \\
\hline
\end{tabular}
\label{tab:gekf_ekf_comparison}
\end{table}

\section{Conclusion and Future Work}\label{sec:conclusion}

In this work, we identified a sensor and state estimation model that captures state-dependent measurement noise exhibited by complex sensor modalities such as vision-based LECs. We deployed belief control barrier functions with a coherent risk measure to capture the expected worst-case cost of probabilistic safety constraint violations. Through simulations, we demonstrated that our approach yields less conservative control and higher tracking accuracy, as compared to belief-space estimation and control with a fixed additive noise model. We hope this establishes a foundation for principled modeling of complex sensor noise in safe control frameworks. In future, we aim to investigate sensor–estimator pairs with richer noise models to broaden the applicability of this approach, and validate the accuracy of these models in photorealistic simulation and hardware experiments.

\section{Acknowledgments}

\thanks{The NASA University Leadership initiative (grant \#80NSSC20M0163) provided funds to assist the authors with their research, but this article solely reflects the opinions and conclusions of its authors and not any NASA entity.}


\newpage

\bibliographystyle{IEEEtran}
\bibliography{bibliography}

\end{document}